\documentclass[conference,a4paper]{IEEEtran}
\IEEEoverridecommandlockouts
\usepackage{xcolor}
\usepackage{amsmath}
\usepackage{amssymb}
\usepackage{amsfonts}
\usepackage{amscd}
\usepackage{amsthm}
\usepackage[noadjust]{cite}
\usepackage{textcomp}

\usepackage{mathrsfs}
\usepackage{amsbsy}

\usepackage{xfrac}

\newtheorem{theorem}{Theorem}
\newtheorem{lemma}{Lemma}

\newtheorem{corollary}{Corollary}

\theoremstyle{definition}

\newtheorem{definition}{Definition}

\makeatletter
\renewcommand*\env@matrix[1][c]{\hskip -\arraycolsep
  \let\@ifnextchar\new@ifnextchar
  \array{*\c@MaxMatrixCols #1}}
\makeatother
\usepackage{graphicx}

\usepackage{etoolbox}
\AtEndEnvironment{example}{\null\hfill\IEEEQED}%

\newcommand{\p}{\pmb}
\newcommand{\xb}{\pmb{x}} 
\newcommand{\cb}{\pmb{c}}

\newcommand{\Ac}{\mathcal{A}} 
 
\newcommand{\Vc}{\mathcal{V}} 
\newcommand{\Ec}{\mathcal{E}}
\newcommand{\Bc}{\mathcal{B}}
\newcommand{\Uc}{\mathscr{U}} 
\newcommand{\Pc}{\mathscr{P}} 

\newcommand{\Ef}{\mathfrak{E}}
\newcommand{\Df}{\mathfrak{D}}
\newcommand{\mais}{{\sf MAIS}}
\newcommand{\Bfull}{{\Bc^{*}}}
\newcommand{\BfullS}{{\Bc_{\!S}^{*}}}
\newcommand{\Gfull}{{G^{*}}}
\newcommand{\Gufull}{{G_{\!u}^{*}}}
\newcommand{\GfullS}{{G_{\!S}^{*}}}
\newcommand{\GufullS}{{G_{\!S,u}^{*}}}
\newcommand{\GfullM}{{G_{\!M}^{*}}}
\newcommand{\GfullMdash}{{G_{\!M'}^{*}}}

\newcommand{\betaoptG}{{\beta_{{\sf opt},G}}}
\newcommand{\betaoptB}{{\beta_{{\sf opt},\Bc}}}
\newcommand{\betaoptGfull}{{\beta_{{\sf opt},\Gfull}}}

\begin{document}
\title{Index Codes with Minimum Locality for\\ Three Receiver Unicast Problems} 
\author{\IEEEauthorblockN{Smiju Kodamthuruthil Joy and Lakshmi Natarajan}
\thanks{The authors are with Department of Electrical Engineering, Indian Institute of Technology Hyderabad, Sangareddy 502\,285, India (email: \{ee17resch11017,\,lakshminatarajan\}@iith.ac.in).}
}%

\maketitle

\begin{abstract}
An index code for a broadcast channel with receiver side information is \emph{locally decodable} if every receiver can decode its demand using only a subset of the codeword symbols transmitted by the sender instead of observing the entire codeword.
Local decodability in index coding improves the error performance when used in wireless broadcast channels, reduces the receiver complexity and improves privacy in index coding.
The \emph{locality} of an index code is the ratio of the number of codeword symbols used by each receiver to the number message symbols demanded by the receiver. 
Prior work on locality in index coding have considered only single unicast and single-uniprior problems, and the optimal trade-off between broadcast rate and locality is known only for a few cases. 
In this paper we identify the optimal broadcast rate (including among non-linear codes) for all three receiver unicast problems when the locality is equal to the minimum possible value, i.e., equal to one.
The index code that achieves this optimal rate is based on a clique covering technique and is well known.
The main contribution of this paper is in providing tight converse results by relating locality to broadcast rate, and showing that this known index coding scheme is optimal when locality is equal to one.
Towards this we derive several structural properties of the side information graphs of three receiver unicast problems, and combine them with information theoretic arguments to arrive at a converse.
\end{abstract}

\section{Introduction} \label{sec:introduction}

Index coding is a class of network coding problems with a single broadcast link connecting a transmitter with multiple receivers~\cite{BBJK_IT_11,BBJK_FOCS_06}. Each receiver or user demands a subset of messages available at the transmitter while knowing another subset of messages as side information.
The code design objective is to broadcast a codeword with as small a length as possible to meet the demands of all the users simultaneously.
The \emph{broadcast rate} or the \emph{rate} of an index code is the ratio of the code length to the length of each of the messages.
The problem of designing index codes with smallest possible broadcast rate is significant because of its applications, such as multimedia content delivery~\cite{BiK_INFOCOM_98}, coded caching~\cite{MaN_IT_14}, distributed computation~\cite{LLPPR_IT_18}, and also because of its relation to network coding~\cite{RSG_IEEE_IT_10,ERL_IT_15} and coding for distributed storage~\cite{Maz_ISIT_14,ShD_ISIT_14}.

An index code is \emph{locally decodable} if every user can decode its demand by using its side information and by observing only a subset of the transmitted codeword symbols (instead of observing the entire codeword)~\cite{HaL_ISIT12}. 
The \emph{locality} of an index code is the ratio of the maximum number of codeword symbols observed by any receiver to the number of message symbols demanded by the receiver~\cite{NKL_ISIT18}. The objective of designing locally decodable index codes is to construct coding schemes that simultaneously minimize rate and locality, and attain the optimal trade-off between these two parameters.
Locally decodable index codes reduce the number of transmissions that any receiver has to listen to, and hence, reduce the receiver complexity as well, see for instance~\cite{VaR_arxiv_19}.
When index codes are to be used in a fading wireless broadcast channel, the probability of error at the receivers can be reduced by using index codes with small locality~\cite{TRAR_TVT_17}.  
Locality is also known to be related to privacy in index coding~\cite{KSCF_arxiv_18}.

We consider \emph{unicast index coding problems} where the transmitter is required to broadcast $N$ independent messages $\xb_1,\xb_2\dots,\xb_N$ to $n$ users or receivers $u_1,\dots,u_n$. Each message $\xb_j$, $j\in[N]$, is desired at exactly one of the receivers, while each receiver can demand any number of messages. 
The receiver $u_i$ wants the messages $\xb_j$, $j \in W_i$, and knows $\xb_j$, $j \in K_i$ as side information, where $W_i,K_i \subset [N]$ and $W_i \cap K_i = \phi$. Since we consider only unicast problems, we have $W_i \cap W_j = \phi$ for $i \neq j$. The set $Y_i = [N] \setminus (W_i \cup K_i)$ corresponds to the messages that $u_i$ neither demands nor knows.
Note that $n \leq N$ and equality holds if and only if every receiver demands a unique message from the transmitter. The case \mbox{$n=N$} is known as \emph{single unicast index coding}.

Prior work on locally decodable index codes have considered single unicast index coding problems~\cite{HaL_ISIT12,NKL_ISIT18,VaR_arxiv_19}, problems where each receiver demands exactly one message~\cite{KSCF_arxiv_18}, and a family of index coding problems called \emph{single uniprior}\footnote{Single uniprior problems are index coding problems where every message is available as side information at a unique receiver.}~\cite{TRAR_TVT_17,OHL_IT_16}.
The exact characterization of the optimal trade-off between rate and locality is known in only a few cases, all of them being single unicast problems: optimal rate among non-linear codes for single unicast problems and locality equal to one~\cite{NKL_ISIT18}, rate-locality trade-off among linear codes for single-unicast problems when min-rank is one less than the number of receivers~\cite{NDKL_arxiv_19}.

In this paper, we consider unicast index coding problems with three or fewer receivers for the smallest possible value of locality, i.e, locality equal to one.
We show that for any such problem, the index coding scheme based on clique covering of a related graph (the underlying undirected side information graph) provides the optimal rate among all codes, including non-linear codes.  
This explicitly identifies one end-point of the optimal trade-off between rate and locality, namely the minimum-locality point, for all unicast index coding problems with three or fewer receivers.
The clique covering based coding scheme that achieves this point on the optimal rate-locality trade-off is well known. The main contribution of this paper is in proving converse results, i.e., showing that this scheme is optimal among all codes with locality equal to one, including non-linear codes. 

Our converse relies on several properties of the directed and undirected side information graphs of the three receiver unicast problems. We derive these technical results, which are related to independence number, perfectness and maximum acyclic induced subgraphs, in Section~\ref{sec:technical_prelim}. 
These results are combined with graph theoretic and information theoretic arguments in Section~\ref{sec:main_result} to prove the main result of this paper.
The system model and related background are reviewed in Section~\ref{sec:defn}.

\emph{Notation:} For any positive integer $N$, $[N]$ denotes the set $\{1,\dots,N\}$. The symbol $\phi$ denotes the empty set. Vectors are denoted using bold small letters, such as $\xb$.

\section{System Model and Background} \label{sec:defn}

\subsubsection*{Graphs associated with index coding}

We represent unicast and single unicast index coding problems using (directed) bipartite graphs and (directed) side information graphs, respectively. 
Following~\cite{TDN_ISIT12} we represent an unicast index coding problem by a directed bipartite graph $\Bc=(\Uc,\Pc,E)$ where $\Uc=\{u_1,\dots,u_n\}$ is the vertex set of all receivers, and $\Pc=\{\xb_1,\dots,\xb_N\}$ is the vertex set of messages. 
The edge set $E$ contains $(\xb_j,u_i)$ if $j \in W_i$ and contains $(u_i,\xb_j)$ if $j \in K_i$. 

When the unicast index coding problem is also single unicast, i.e., when each message is demanded by a unique receiver, we use the side information graph~\cite{BBJK_IT_11} $G=(\Vc,\Ec)$ to represent the problem. 
Here, the vertex set $\Vc=[N]$, and the edge set $\Ec$ contains the directed edge $(i,j)$ if the receiver demanding $\xb_i$ knows the message $\xb_j$ as side information.
The \emph{underlying undirected side information graph}~\cite{SDL_ISIT13} \mbox{$G_u=(\Vc,\Ec_u)$} corresponding to $G$ is the graph with vertex set $\Vc=[N]$ and an undirected edge set \mbox{$\Ec_u=\left\{\,\{i,j\} \, | \, (i,j),\,(j,i) \in \Ec \right\}$}, i.e., \mbox{$\{i,j\} \in \Ec_u$} if and only if 
both $(i,j), (j,i) \in \Ec$.

\subsubsection*{Locality in index coding}

We assume that the messages $\xb_1,\dots,\xb_N$ are vectors over a finite alphabet $\Ac$, i.e., $\xb_i \in \Ac^m$ for some integer $m$. The encoder $\Ef$ at transmitter maps $\xb_1,\dots,\xb_N$ into a length $\ell$ codeword $\cb \in \Ac^{\ell}$.
We assume that each receiver observes only a subset of the codeword symbols. Specifically, let $u_i$ observe the subvector $\cb_{R_i}=(c_k,k \in R_i)$ where $R_i \subseteq [\ell]$. 
The decoder $\Df_i$ at $u_i$ outputs the demand $\xb_{W_i}=(\xb_j,j \in W_i)$ using the channel observation $\cb_{R_i}$ and side information $\xb_{K_i}=(\xb_j,j \in K_i)$ as inputs. 
We say that $(\Ef,\Df_1,\dots,\Df_n)$ is a \emph{valid index code} if every receiver can decode its demand using its side information and channel observation. 
The \emph{broadcast rate} of the index code is $\beta=\ell/m$. Note that $u_i$ observes $|R_i|$ coded symbols to decode $m|W_i|$ message symbols present in the vector $\xb_{W_i}$. 
The \emph{locality} at $u_i$ is $r_i=|R_i|\,/\,m|W_i|$ and the \emph{locality} or the \emph{overall locality of the index code} is $r=\max_{i \in [n]} r_i$. 
The locality of the index code is the maximum number of channel observations made by any receiver to decode one message symbol.
Since the side information at $u_i$ is independent of the demanded message $\xb_{W_i}$, the receiver must observe at least $m|W_i|$ coded symbols to be able to decode $\xb_{W_i}$, i.e., $r_i = |R_i|\,/\,m|W_i| \geq 1$. Thus, $r \geq 1$ for any valid index code. We say that an index code has minimum locality if $r=1$.

The optimal broadcast rate (infimum among the rates of all valid index codes, without any restriction on the locality $r$ of the codes or on the message length $m$) of a unicast index coding problem $\Bc$ will be denoted by $\betaoptB$, and that of a single unicast problem $G$ by $\betaoptG$.

\begin{definition}
For a unicast index coding problem $\Bc$ (respectively, for a single unicast problem $G$), the {\em optimal broadcast rate function} $\beta_{\Bc}^*(r)$ (respectively $\beta_{G}^*(r)$) is the infimum of the broadcast rates among all valid index codes over all possible message length $m \geq 1$ with locality at the most $r$.
\end{definition}

Note that the function $\beta_G^*(r)$ is non-increasing and $\beta_G^*(r) \geq \betaoptG$ for any $r$, since $\betaoptG$ is the best broadcast rate achievable without any constraints on the locality of the code.

\subsubsection*{Graph-theoretic background}

We will briefly recall relevant graph-theoretic terminology~\cite{graphtheory}. 
A subset $S$ of the vertices of an undirected graph $G_u$ is an {independent set} if no two vertices in $S$ are adjacent. 
The number of vertices in a maximum-sized independent set of $G_u$ is called the independence number of $G_u$ and is denoted by $\alpha(G_u)$. 
For any directed graph $G$, let $\mais(G)$ denote the size of the maximum acyclic induced subgraph of $G$.
If $S$ is a subset of vertices of $G$, let $G_S$ denote the subgraph of $G$ induced by $S$, i.e., $G_S$ has vertex set $S$ and $G_S$ consists of all edges in $G$ with both end points in $S$.

An undirected graph $G_u$ is called \emph{perfect} if for every induced subgraph $H$, $\alpha(H) = \bar{\chi}(H)$, where $\bar{\chi}$ denotes the clique covering number. 
For any $G_u$, $\alpha(G_u)\leq \bar{\chi}_f(G_u)\leq \bar{\chi}(G_u)$, where $\bar{\chi}_f$ denotes the fractional clique covering number. Thus for a perfect graph we have $\bar{\chi}_f(G_u)= \bar{\chi}(G_u)$.

\begin{theorem}[{The strong perfect graph theorem~\cite{strongperfect}}] \label{thm:strong_perfect}
An undirected graph $G_u$ is perfect if and only if no induced subgraph of $G_u$ is an odd hole (odd cycle of length at least 5) or an odd antihole (compliment of odd hole).
\end{theorem}

It is well known that for any single unicast problem $G$, $\betaoptG \geq \mais(G)$.
The optimal rate for locality $r=1$ is 
\begin{equation} \label{eq:rate_clique_covering}
\beta_G^*(1) = \bar{\chi}_f(G_u), 
\end{equation} 
see~\cite{NKL_ISIT18,HaL_ISIT12}, where $G_u$ is the underlying undirected graph corresponding to $G$.

\section{Technical Preliminaries} \label{sec:technical_prelim}

We now identify key properties of the graphs associated with the unicast index coding problems. These results will be vital in identifying the optimal minimum-locality index codes for three or fewer receivers.

\begin{lemma}\label{lma:mais_ind_set}
For any single unicast problem $G$, $\mais(G) \leq \alpha(G_u)$, where $G_u$ is the underlying undirected graph.
 \end{lemma}
\begin{proof}
Consider any subset $S$ of vertices of $G=(\Vc,\Ec)$ such that $G_S$ is acyclic. Clearly, for any $i,j \in S$, $\Ec$ can not simultaneously contain both $(i,j)$ and $(j,i)$, since this implies the existence of a length $2$ cycle in $G_S$. Thus, $\{i,j\}$ is not an edge in $G_u$. Thus, if $S$ is such that $G_S$ is acyclic, then $S$ is an independent set in $G_u$.
Hence, $\mais(G)\leq \alpha(G_u)$.
\end{proof}

Equality holds in Lemma~\ref{lma:mais_ind_set} if and only if there exists a largest independent set $S$ of $G_u$ such that $G_S$ is acyclic.

\subsection{Equivalent Single Unicast Problem}

For any given unicast index coding problem $\Bc$, we consider a corresponding single unicast problem $G$ which we refer to as the \emph{equivalent single unicast problem (ESUP)} of $\Bc$.
The ESUP of a unicast problem $\Bc$ is constructed by using the following well known procedure~\cite{BBJK_IT_11}. Suppose the user $u_i$ in $\Bc$ has want set $W_i=\{i_1,\dots,i_{|W_i|}\}$ and side information set $K_i$.
Corresponding to each user $u_i$ in $\Bc$, the ESUP contains $|W_i|$ receivers all equipped with the same side information $\xb_{K_i}$, and these $|W_i|$ receivers demand one message each, $\xb_{i_1},\dots,\xb_{i_{|W_i|}}$, respectively. Thus the ESUP consists of $N=\sum_{i \in [n]}|W_i|$ receivers and an equal number of messages. 


\begin{theorem} \label{thm:U_ESU}
For any unicast index coding problem $\Bc$ and its ESUP $G$, we have
$\beta_G^*(\max_i{|W_i|})\leq\beta_{\mathcal{B}}^\ast(1)\leq\beta_{G}^\ast(1)$. 
\end{theorem}
\begin{proof}
Will will first prove the second inequality. Assume a valid index code with locality $r=1$ and message length $m$ for $G$. Since $r=1$ every receiver in $G$ observes exactly $m$ coded symbols to decode its demand. 
The $|W_i|$ receivers in $G$ corresponding to the user $u_i$ in $\Bc$ will together observe at the most $m|W_i|$ coded symbols to decode the $|W_i|$ messages $\xb_{i_1},\dots,\xb_{i_{|W_i|}}$, where $W_i=\{i_1,\dots,i_{|W_i|}\}$.
Thus, using the same index code for $\Bc$, $u_i$ needs to observe $m|W_i|$ codeword symbols to decode its demand $\xb_{W_i}$, yielding locality $r_i=1$.
This is true for every $i \in [n]$.
Hence any valid code for $G$ with $r=1$ is also a valid code for $\Bc$ with $r=1$. Therefore, $\beta_{\mathcal{B}}^\ast(1)\leq\beta_{G}^\ast(1)$.

Next we consider a valid code for $\Bc$ with $r=1$. 
Here $u_i$ uses $|R_i|=m|W_i|$ codeword symbols and the side information $\xb_{K_i}$ to decode its demand $\xb_{W_i}$.
Using the same index code in $G$, we note that each of the $|W_i|$ users in $G$ corresponding to $u_i$, can decode their respective demands from $m|W_i|$ codeword symbols and the common side information $\xb_{K_i}$. 
Since each of these receivers in $G$ demands exactly one message, their localities are equal to $m|W_i|/m = |W_i|$.
Considering all the receivers in $G$, the overall locality $r=\max_{i \in [n]}r_i$. Thus any index code with $r=1$ for $\Bc$ is also a valid index code for $G$ with $r=\max_i |W_i|$.
Therefore, $\beta_{\mathcal{B}}^\ast(1)\geq\beta_{G}^\ast(\max_i{|W_i|})$. 
\end{proof}

\subsection{The Three Receiver Unicast Problem $\Bfull$} \label{sec:Bfull}

In this paper we are interested in unicast problems with three or fewer receivers and where all messages are to be encoded at the same rate, i.e., we assume all message vectors $\xb_1,\dots,\xb_N$ have the same length.
We now consider a specific unicast index coding problem with \mbox{$n=3$} receivers whose bipartite graph will be denoted as $\Bfull$. 
Any other unicast problem with three or fewer receivers can be identified as a sub-problem of $\Bfull$.

Since we have \mbox{$n=3$}, the index set of all messages $[N]=W_1\cup W_2\cup W_3=W_i\cup K_i\cup Y_i$ for every $i \in [3]$, where $Y_i=[N] \setminus (W_i \cup K_i)$ is the index set of \emph{interference messages}, i.e., messages which are not demanded and are not known by $u_i$. 
For any $i\neq j$, we have $|W_i\cap W_j|=0$ and $W_i\subset (K_j\cup Y_j)$. 
Assuming \mbox{$i\neq j \neq k \neq i$}, the index set $W_i$ can be partitioned into the following $4$ disjoint subsets, \mbox{$W_i\cap K_j\cap K_k$}, \mbox{$W_i\cap Y_j\cap Y_k$}, \mbox{$W_i\cap K_j\cap Y_k$} and \mbox{$W_i\cap K_k\cap Y_j$}.
For example, any message with its index in the set \mbox{$W_i\cap K_j\cap Y_k$} is demanded by $u_i$, is known to $u_j$ as side information, and is neither wanted and nor known at $u_k$.
Since all messages with indices in \mbox{$W_i \cap K_j \cap Y_k$} can be viewed as one single message, we will assume that \mbox{$|W_i \cap K_j \cap Y_k|$} is either $0$ or $1$.
Thus, $|W_i| \leq 4$ for each $i \in [3]$, and in all, a three receiver unicast problem consists of $12$ disjoint subsets of messages, where the size of each subset is either $0$ or $1$.

\begin{figure}[t]
    \centering
    \includegraphics[width=3in]{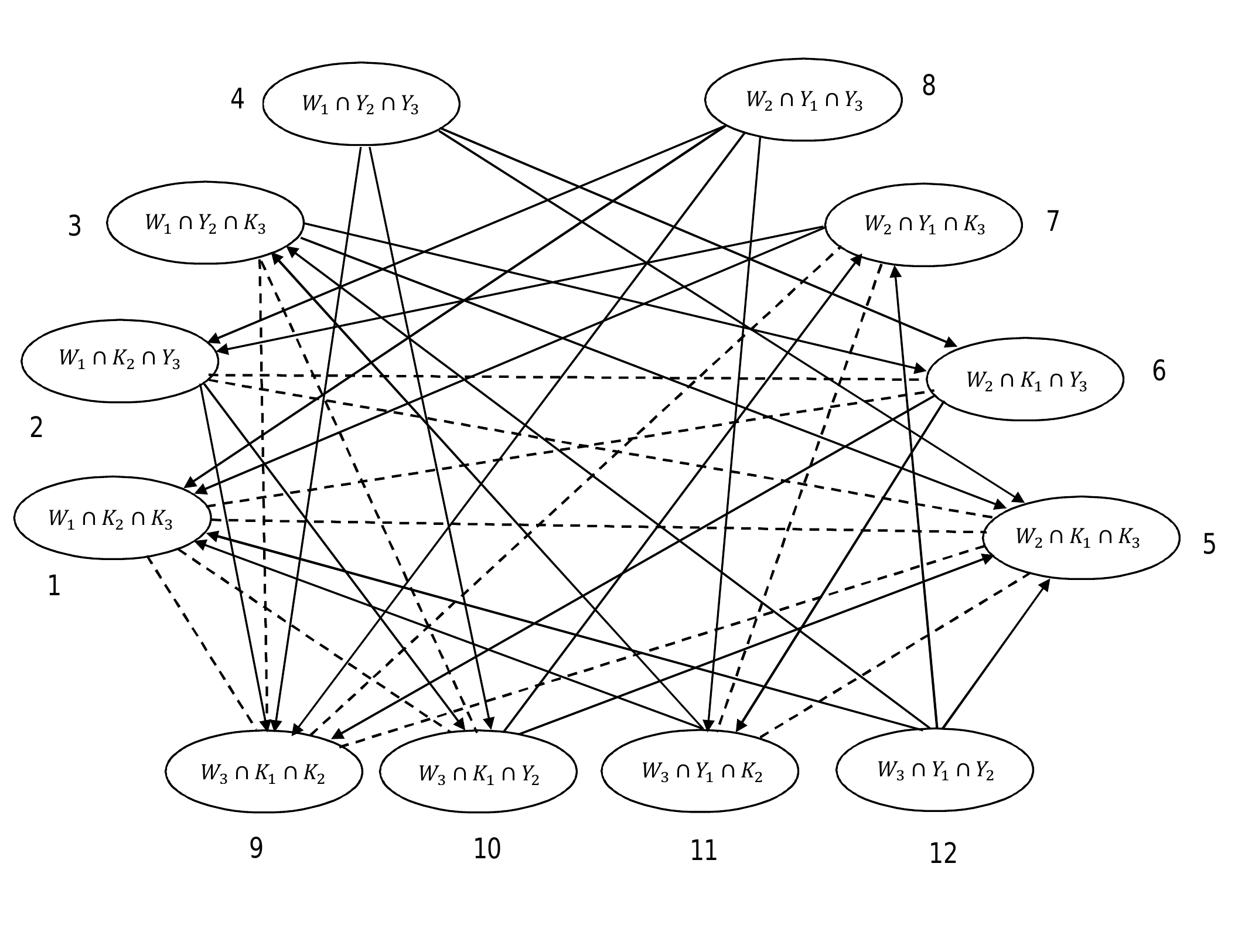}
    \vspace{-5mm}
    \caption{Side information graph $\Gfull$. Dashed lines represent two directed edges in either direction. For example, both $(1,5)$ and $(5,1)$ are edges in $\Gfull$.}
    \label{fig:side_graph}
\end{figure}

\begin{figure}[t]
    \centering
    \includegraphics[width=3in]{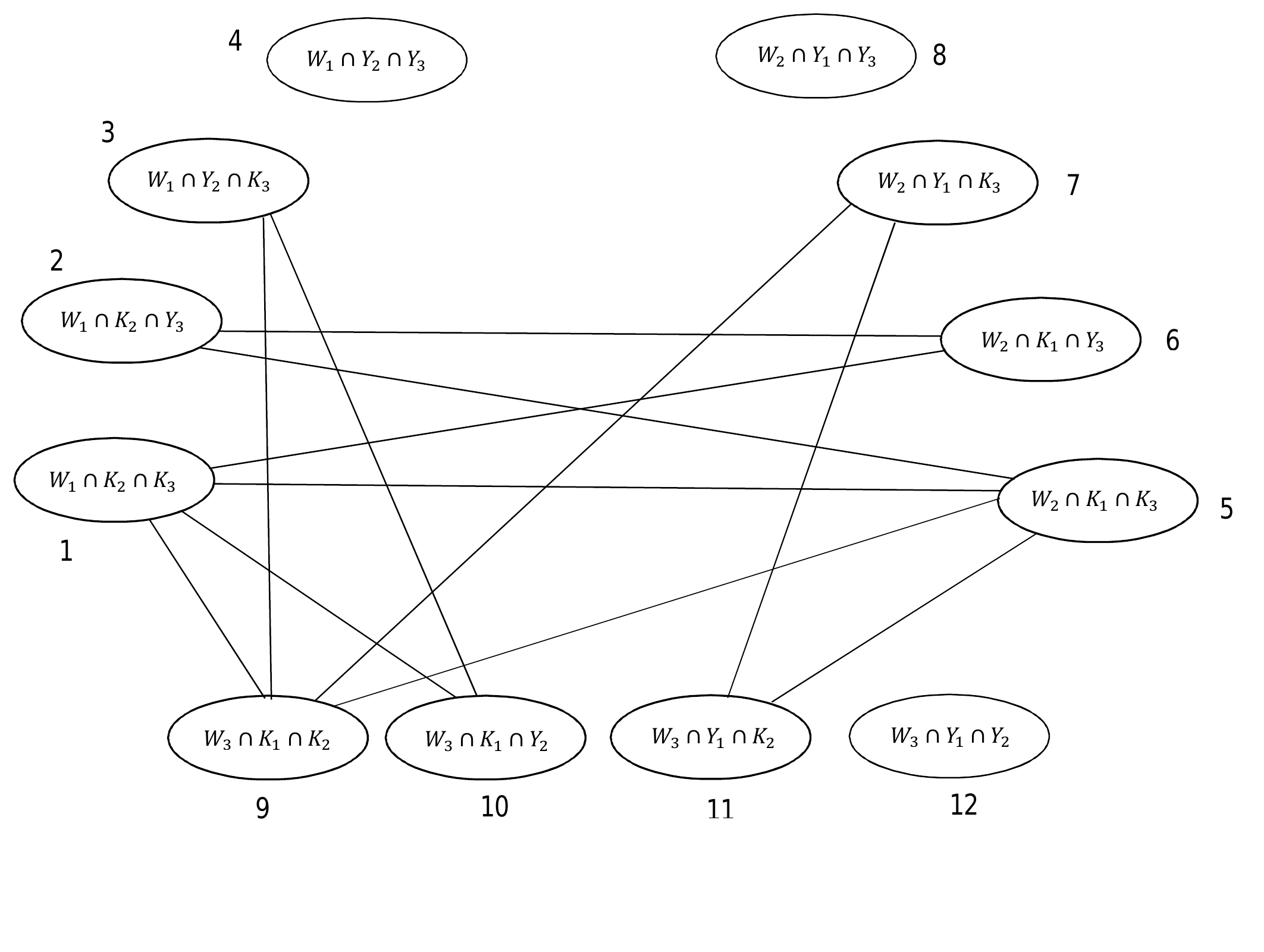}
    \vspace{-5mm}
    \caption{The underlying undirected side information graph $\Gufull$.}
    \label{fig:Gu}
\end{figure}

Now we consider the specific three receiver unicast problem $\Bfull$ where the size of each of the $12$ subsets of messages is equal to $1$, i.e., $N=12$.
Here each receiver demands 4 messages.
Let $\xb_{W_1}=(\xb_1,\xb_2,\xb_3,\xb_4$), $\xb_{W_2}=(\xb_5,\xb_6,\xb_7,\xb_8)$ and $\xb_{W_3}=(\xb_9,\xb_{10},\xb_{11},\xb_{12})$. 
The ESUP consists of $12$ receivers each demanding a unique message.
The side information graph $\Gfull$ of this ESUP and the underlying undirected graph $\Gufull$ are shown in Fig.~{\ref{fig:side_graph} and~\ref{fig:Gu}}, respectively. 
We represent each vertex or message interchangeably by the message index or by the subset corresponding to that index as shown in Fig.~\ref{fig:side_graph} and~\ref{fig:Gu}.
The construction of $\Gfull$ is as follows. 
Note that there are $4$ messages, and correspondingly $4$ vertices, in $\Gfull$ associated with each of $W_1,W_2,W_3$.
We represent each message or vertex of $\Gfull$ using the subset corresponding to that message. The subsets are of one of the following three types: \emph{(i)}~$W_i\cap K_j\cap Y_k$, \emph{(ii)}~$W_i\cap K_j\cap K_k$, \emph{(iii)}~$W_i\cap Y_j\cap Y_k$. 
The subset $W_i\cap K_j\cap Y_k$ corresponds to the message in $W_i$ which is known to all the $4$ receivers of $\Gfull$ corresponding to the $4$ subsets of $W_j$, and is an interference to all the $4$ receivers corresponding to the $4$ subsets of $W_k$.
Therefore the vertex $W_i\cap K_j\cap Y_k$ has incoming edges from all the $4$ vertices corresponding to the $4$ subsets of $W_j$, and no incoming edges from any other vertices.
For example, vertex $2$( which corresponds to $W_1\cap K_2\cap Y_3$) has incoming edges from vertices $5,6,7,8$ (these correspond to subsets of $W_2$). 
Similarly, the vertices of the type $W_i \cap K_j \cap K_k$ have $8$ incoming edges, and $W_i \cap Y_j \cap Y_k$ have no incoming edges.

For any two vertices $i$ and $j$ in $\Gfull$, $i,j \in [12]$, we say that there is a \emph{bidirectional edge} between $i$ and $j$ in $\Gfull$ if both $(i,j)$ and $(j,i)$ are edges in $\Gfull$. Note that there is a bidirectional edge between $i$ and $j$ in $\Gfull$ if and only if $\{i,j\}$ is an edge in $\Gufull$.

\begin{lemma}\label{lma:bidir_edges}
Let $(i,j,k)$ be any permutation of $(1,2,3)$. In the side information graph $\Gfull$, we have
\begin{enumerate}
\item[\emph{(i)}] $W_i\cap K_j\cap Y_k$ forms bidirectional edges with the subsets $W_j\cap K_i\cap Y_k$ and $W_j\cap K_i\cap K_k$; 
\item[\emph{(ii)}] $W_i\cap K_j\cap K_k$ forms bidirectional edges with the subsets $W_j\cap K_i\cap Y_k$ , $W_k\cap K_i\cap Y_j$,$W_j\cap K_i\cap K_k$ and $W_k\cap K_i\cap K_j$; 
\item[\emph{(iii)}] No bidirectional edge is incident on $W_i\cap Y_j\cap Y_k$.
\end{enumerate}  
\end{lemma}
\begin{proof}
\emph{(i)}~Vertex $W_i\cap K_j\cap Y_k$ has incoming edges from all $4$ vertices corresponding to the $4$ subsets of $W_j$ and has outgoing edges to all vertices that are subsets of $K_i$. 
Therefore $W_i\cap K_j\cap Y_k$ forms bidirectional edge with  the subsets $W_j\cap K_i\cap Y_k$ and $W_j\cap K_i\cap K_k$.

\emph{(ii)}~$W_i\cap K_j\cap K_k$  has incoming edges from all subsets of $W_j$ and $W_k$ and outgoing edges to all vertices that are subsets of $K_i$. Thus $W_i\cap K_j\cap K_k$ forms bidirectional edge with  the subsets $W_j\cap K_i\cap Y_k$, $W_k\cap K_i\cap Y_j$, $W_j\cap K_i\cap K_k$ and $W_k\cap K_i\cap K_j$

\emph{(iii)}~$W_i\cap Y_j\cap Y_k$ has no incoming edges since this message is not available as side information at any receiver.  
\end{proof}

\begin{theorem}\label{thm:Gu_perfect}
The undirected graph $\Gufull$ is perfect.
\end{theorem}
\begin{proof}
We use Theorem~\ref{thm:strong_perfect} (the strong perfect graph theorem) to prove $\Gufull$ is perfect. 
By Lemma~\ref{lma:bidir_edges} and Fig.~\ref{fig:Gu}, we have the following observation: vertices $4,8,12$ have no edges incident on them, $2,3,6,7,10,11$ have degree $2$, and the remaining vertices $1,5,9$ have degree $4$ and they form a clique $C=\{1,5,9\}$. Note that every degree $2$ vertex is adjacent to another degree $2$ vertex and one vertex from $C$.

We first show that $\Gufull$ contains no odd hole. 
Towards that, clearly $4,8,12$ can not be a part of an odd hole since their degrees are zero. 
Suppose a degree two vertex $i$ is a part of a cycle, then the cycle must contain both the neighbors of this degree two vertex: a vertex $j$ of degree two and a vertex $k \in C$. 
Since any degree two vertex is adjacent to a vertex in $C$, there exists an $l \in C$ such that $j$ and $l$ are neighbors. Since $k,l \in C$, $k$ and $l$ are adjacent as well. 
Thus, we conclude that $i\!-j\!-l\!-k\!-i$ forms a cycle of length $4$. Thus, a cycle containing any vertex from $\{2,3,6,7,10,11\}$ can not be an odd hole. 
Clearly, the remaining three vertices $1,5,9$ form a cycle of length $3$, which is not an odd hole. 
  
We now show that $\Gufull$ does not contain an odd antihole.
Since the odd antihole of length $5$ is isomorphic to the odd hole of length $5$, and since we already know that $\Gufull$ does not contain any odd hole, we conclude that $\Gufull$ does not contain the length $5$ odd antihole.
Any odd antihole of length $7$ or more, contains at least $7$ vertices each of degree at least $4$\footnote{Every vertex in an antihole of size $n$ is adjacent to exactly $n-3$ vertices.}.
Since $\Gufull$ contains only three vertices with degree $4$ or more, we conclude that $\Gufull$ does not contain any odd antihole.
\end{proof}

For any $S \subseteq [12]$, let $\GfullS$ and $\GufullS$ be the subgraphs of $\Gfull$ and $\Gufull$, respectively, induced by the vertices $S$.
Note that $\GufullS$ is the underlying undirected graph corresponding to $\GfullS$.

\begin{figure}[t]
    \centering
    \includegraphics[width=3in]{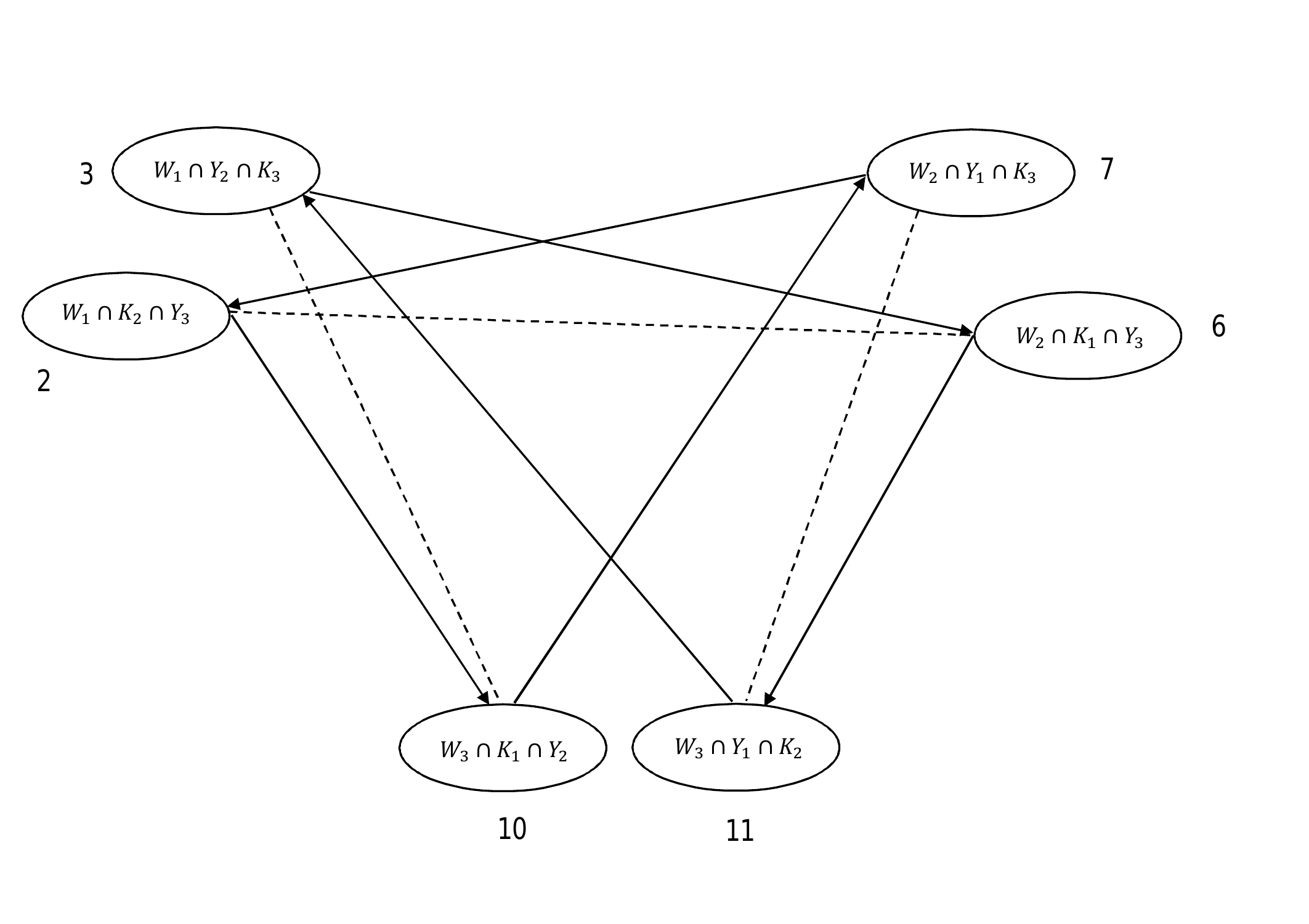}
    \vspace{-5mm}
    \caption{The two directed cycles (of length 3) in $\Gfull$ and the edges induced by these vertices. Dashed lines represent bidirectional edges.}
    \label{fig:2_dir_cycles}
\end{figure}
\begin{lemma}\label{lma:2_cycles_G}
If $S$ is an independent set in $\Gufull$ and $\GfullS$ is a directed cycle, then either $S=\{2,7,10\}$ or $S=\{3,6,11\}$.
\end{lemma}
\begin{proof}
Since $S$ is an independent set in $\Gufull$,  
it is clear that there are no bidirectional edges in $\GfullS$, and the vertices in $S$ are connected by unidirectional edges only. 
Thus, in the following, we do not consider bidirectional edges present in $\Gfull$.
Vertices that have either only incoming or only outgoing edges can not be part of a cycle, i.e., from Fig.~\ref{fig:side_graph}, $1,4,5,8,9,12 \notin S$.
The rest of the vertices and the edges among them are shown in Fig.~\ref{fig:2_dir_cycles}. 
From Fig.~\ref{fig:2_dir_cycles}, clearly there are exactly $2$ cycles of length $3$ contained by these vertices, viz.\ $(2,10,7)$ and $(3,6,11)$.
\end{proof}

We know from Lemma~\ref{lma:mais_ind_set}, that $\mais(\GfullS) \leq \alpha(\GufullS)$ for any choice of $S$.
We now characterize subsets $S$ for which the maximum acyclic induced subgraph of $\GfullS$ is strictly smaller than the independence number of $\GufullS$. 
This result will be used in converse arguments in the next section. 
Note that the vertices $4,8,12$ are of degree zero in $\Gufull$. 

\begin{theorem}\label{thm:mais_less_alpha}
For any \mbox{$S \subseteq [12]$} and \mbox{$I=\{4,8,12\}\cap S$},
\mbox{$\mais(\GfullS)< \alpha(\GufullS)$} if and only if either \mbox{$S=\{2,10,7\}\cup I$} or \mbox{$S=\{3,6,11\}\cup I$}. 
\end{theorem}
\begin{proof}
\emph{The `only if' part:}
Let \mbox{$\mais(\GfullS)< \alpha(\GufullS)$}. 
Then there exists a maximum-sized independent set \mbox{$M \subseteq S$} of $\Gufull$ such that the $\GfullM$ contains a cycle and $|M| = \alpha(\GufullS)$. 
From Lemma~\ref{lma:2_cycles_G}, $\{2,10,7\} \subseteq M$ or $\{3,6,11\} \subseteq M$.
Also, exactly one of the following holds: either $\{2,10,7\} \subseteq M$ or $\{3,6,11\} \subseteq M$ since $M$ is independent in $\Gufull$ and $\{2,6\}$, $\{3,10\}$, $\{7,11\}$ are edges in $\Gufull$.
In the rest of the proof, we will assume $\{2,10,7\} \subseteq M$, the proof of the other case is similar.

Since $M$ is independent in $\Gufull$ and $2,10,7 \in M$, from Lemma~\ref{lma:2_cycles_G} and Fig.~\ref{fig:2_dir_cycles}, we conclude that $1,3,5,6,9,11\notin M$, since each of these vertices is adjacent to at least one vertex among $\{2,10,7\}$ in $\Gufull$.
Since $I \subseteq \{4,8,12\}$, no vertex in $I$ is adjacent to $2,7,10$. 
Considering $I \subseteq S$ and the maximality of $M$, we conclude that $M = \{2,10,7\} \cup I$. 
Now, since $\{2,7,10\}$ and $I$ are disjoint, $\alpha(\GufullS) = |M| = 3 + |I|$. From the hypothesis of this theorem, we conclude 
\begin{equation} \label{eq:mais_2+I}
\mais(\GfullS) \, \leq \, \alpha(\GufullS) - 1 \, = \, 2 + |I|.
\end{equation} 

To complete the proof, we will show that $M=S$, i.e., show that $1,3,5,6,9,11 \notin S$. Will prove this by showing a contradiction. Suppose there exists an $i \in \{1,3,5,6,9,11\}$ such that $i \in S$. It can be verified by direct inspection (see Fig.~\ref{fig:Gu}) that every vertex in $\{1,3,5,6,9,11\}$ is adjacent in $\Gufull$ to exactly one vertex in $\{2,7,10\}$. Suppose $j \in \{2,7,10\}$ is adjacent to $i$ in $\Gufull$.
Clearly, 
\begin{equation*}
M' = I \cup \{2,7,10,i\}\!\setminus\!\{j\}=M \cup \{i\} \setminus \{j\} 
\end{equation*} 
is an independent set in $\Gufull$ and $\GufullS$, $|M'|=M=3+|I|$ and $M' \subseteq S$.
Also, neither $\{2,7,10\}$ nor $\{3,6,11\}$ are subsets of $M'$. Hence, from Lemma~\ref{lma:2_cycles_G}, $\GfullMdash$ does not contain cycles. Thus,
$\mais(\GfullS) \geq |M'| = 3 + |I|$, contradicting~\eqref{eq:mais_2+I}. 

\emph{The `if' part:} Suppose $S=\{2,7,10\} \cup I$, where $I \subseteq \{4,8,12\}$. The proof for $S=\{3,6,11\} \cup I$ is similar. From Fig.~\ref{fig:Gu}, it is clear that $\GufullS$ has no edges, i.e., $\alpha(\GufullS) = |S|$. From Fig.~\ref{fig:side_graph}, we observe that $\GfullS$ contains exactly one cycle $(2,10,7)$. Thus, $\mais(\GfullS) = |S|-1 < \alpha(\GufullS)$.
\end{proof}


\begin{corollary} \label{cor:mais_less_alpha}
If $S \subseteq [12]$ is such that $\mais(\GfullS) < \alpha(\GufullS)$, then $S$ is an independent set in $\Gufull$ and $\bar{\chi}(\GufullS) = |S|$.
\end{corollary}
\begin{proof}
Following the proof of the `if' part of Theorem~\ref{thm:mais_less_alpha}, we observe that $S$ is an independent set of $\Gufull$. Consequently, $\GufullS$ has no edges, and hence, $\bar{\chi}(\GufullS) = |S|$.
\end{proof}

\section{Optimal Minimum-Locality Index Codes for Three or Fewer Receivers} \label{sec:main_result}

We will first derive the optimal length of index codes with locality $r=1$ for the unicast index coding problem $\Bfull$, as described in Section~\ref{sec:Bfull}, and then consider all unicast problems with three or fewer receivers. 

\subsection{Minimum-locality index code for $\Bfull$}

Recall that $\Gfull$ is the side information graph of the ESUP of $\Bfull$. 
From the proof of Theorem~\ref{thm:U_ESU}, any valid index code with locality one for $\Gfull$ is also valid for $\Bfull$ with locality one. The optimal broadcast rate with $r=1$ for $\Gfull$ is $\beta_{\Gfull}^*(1) = \bar{\chi}_f(\Gufull)$~\cite{NKL_ISIT18,HaL_ISIT12}. We show that $\beta_{\Bfull}^*(1)$ too is equal to this value.

\begin{theorem}
\label{thm:3_Rxs_sp_case}
For the three receiver unicast problem $\Bfull$ (described in Section~\ref{sec:Bfull}), with the ESUP $\Gfull$ and underlying undirected graph $\Gufull$, 
$\beta_{\Bfull}^\ast(1)= \bar{\chi}_f(\Gufull) = \bar{\chi}(\Gufull)=7$.
\end{theorem}

\subsubsection{Proof of achievability for Theorem~\ref{thm:3_Rxs_sp_case}} \label{sec:3_Rxs_sp_case_achieve}

We know that $\Gufull$ is perfect (Theorem~\ref{thm:Gu_perfect}), and hence, $\bar{\chi}_f(\Gufull) = \bar{\chi}(\Gufull)$. The value $\bar{\chi}_f(\Gufull)=7$ can be verified numerically. Using~\eqref{eq:rate_clique_covering} and Theorem~\ref{thm:U_ESU}, we immediately deduce
\begin{equation*}
\beta_{\Bfull}^*(1) \leq \beta_{\Gfull}^*(1) = \bar{\chi}_f(\Gufull) = \bar{\chi}(\Gufull)=7.
\end{equation*} 

\subsubsection{Proof of converse for Theorem~\ref{thm:3_Rxs_sp_case}} \label{sec:3_Rxs_sp_case_converse}

\begin{figure}[t]
    \centering
    \includegraphics[width=2.75in]{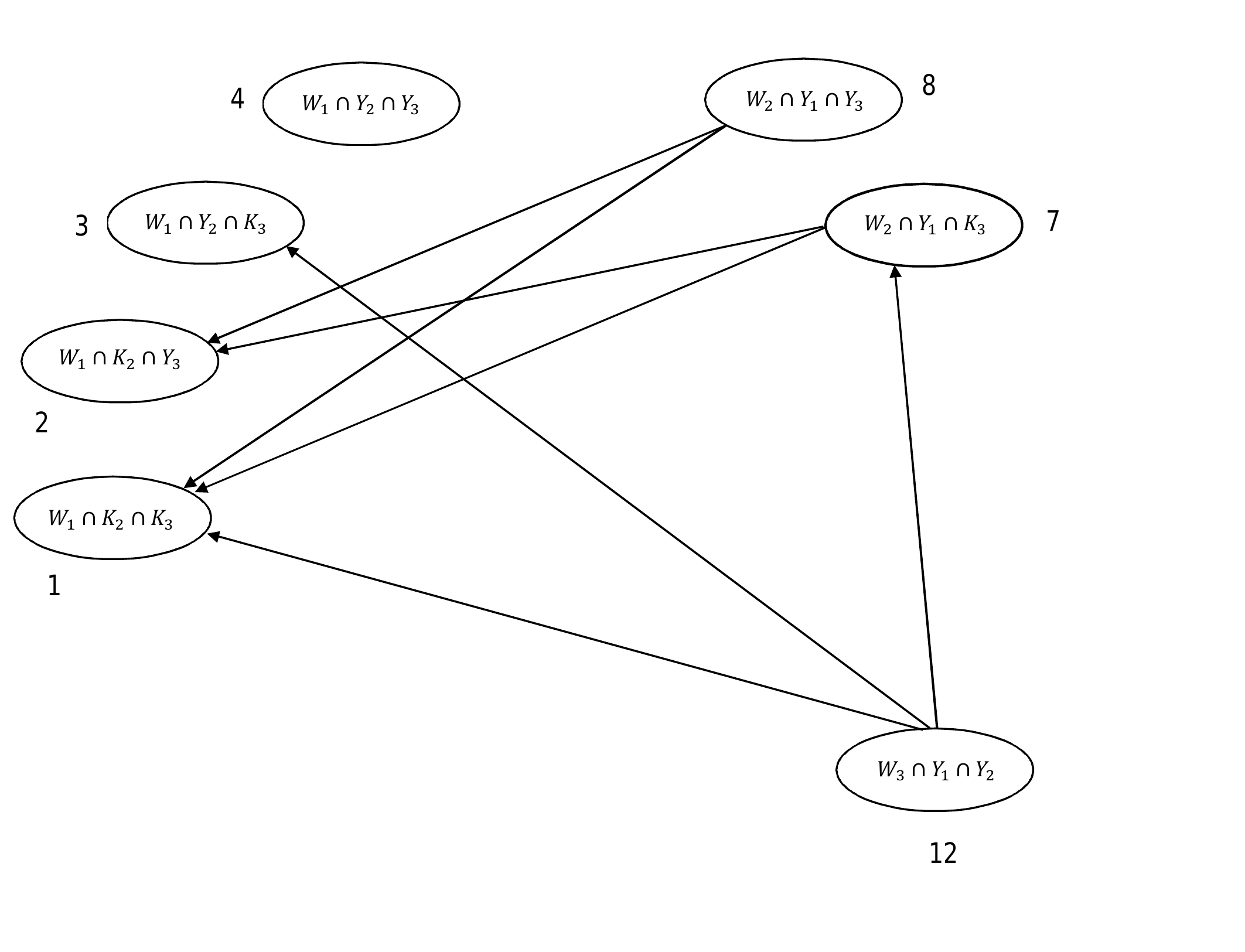}
    \vspace{-5mm}
    \caption{Acyclic subgraph of $\Gfull$ induced by $S=\{1,2,3,4,7,8,12\}.$}
    \label{fig:mais_equal_ind}
\end{figure}

From Theorem \ref{thm:U_ESU}, $\beta_{\Bfull}^*(1) \geq \beta_{\Gfull}^*(\max_i|W_i|)$. 
We know that $\beta_{\Gfull}^*(\max_i|W_i|)\geq \betaoptGfull$, the optimum broadcast rate of the ESUP $\Gfull$. Further, we know that $\betaoptGfull \geq \mais(\Gfull)$. Thus, we have 
\begin{equation} \label{eq:thm:mais_lower_bound}
 \beta_{\Bfull}^*(1) \geq \mais(\Gfull).
\end{equation} 
Now we prove that \mbox{$\mais(\Gfull)=\alpha(\Gufull)$}.
Since $\Gufull$ is perfect, we know that $\alpha(\Gufull)=\bar{\chi}(\Gufull)=7$. 
Consider the directed subgraph of $\Gfull$ induced by $S=\{1,2,3,4,7,8,12\}$, as shown in Fig.~\ref{fig:mais_equal_ind}. 
Since there is only one vertex in the subgraph (vertex 7) with both incoming and outgoing edges, it is clear that this subgraph is acyclic.
Thus, we have $\mais(\Gfull) \geq |S| = 7$.
On the other hand, from Lemma~\ref{lma:mais_ind_set}, we have $\mais(\Gfull) \leq \alpha(\Gufull) = 7$, thus yielding $\mais(\Gfull)=7$.
The converse follows by combining this result with~\eqref{eq:thm:mais_lower_bound}.

\subsection{Minimum-locality codes for all three receiver unicast index coding problems}

From the discussion in Section~\ref{sec:Bfull} we know that any unicast problem involving three (or fewer) receivers consists of $12$ disjoint subsets of messages, with the size of each of these subsets being either $0$ or $1$. 
In other words, with the messages $\xb_1,\dots,\xb_{12}$ as defined in Section~\ref{sec:Bfull} and Fig.~\ref{fig:side_graph}, any three receiver unicast problem is a subproblem of $\Bfull$ induced by a subset $S \subseteq [12]$ of the messages. 
Thus, any three receiver unicast problem can be uniquely identified by its corresponding set of messages $S$.
The bipartite graph of this problem $\BfullS$ consists of the vertex set of users $\Uc=\{u_1,u_2,u_3\}$, the vertex set of messages $\Pc_S=\{\xb_i | i \in S\}$, and all the edges in $\Bfull$ with one of the end points in $\Pc_S$.
It is not difficult to see, that the ESUP of $\BfullS$ is $\GfullS$, the subgraph of $\Gfull$ induced by $S$. Note that the underlying undirected graph is $\GufullS$, and this is the subgraph of $\Gufull$ induced by $S$.
Note that for any unicast problem $\BfullS$, the number of messages in the problem is $|\Pc_S|=|S|$. 
We will now state the main result of this section, which provides the value of $\beta_{\BfullS}^*(1)$ for all three receiver unicast problems $\BfullS$, $S \subseteq [12]$. We provide the proof of this theorem in the rest of the section.

\begin{theorem} \label{thm:main_all_problems}
For any $S \subseteq [12]$, the optimal rate among index codes with locality $1$ for the three receiver unicast problem $\BfullS$ is $\bar{\chi}(\GufullS)$.
\end{theorem}

The proof of achievability of Theorem~\ref{thm:main_all_problems}, i.e., the proof of the claim $\beta_{\BfullS}^*(1) \leq \bar{\chi}(\GufullS)$, is similar to that of Theorem~\ref{thm:3_Rxs_sp_case}.
Since $\Gufull$ is perfect (Theorem~\ref{thm:Gu_perfect}), the vertex induced subgraph $\GufullS$ is perfect as well.
Achievability follows from the perfectness of $\GufullS$, the second inequality in the statement of Theorem~\ref{thm:U_ESU} and~\eqref{eq:rate_clique_covering}.

Our converse for Theorem~\ref{thm:main_all_problems} proceeds considering two cases: \emph{(i)}~\mbox{$\mais(\GfullS) = \alpha(\GufullS)$}, \emph{(ii)}~\mbox{$\mais(\GfullS) < \alpha(\GufullS)$}.
We will now provide the proof of converse for these two cases separately.

\subsection*{Converse for the case $\mais(\GfullS) = \alpha(\GufullS)$}

The converse in this case is similar to that of Theorem~\ref{thm:3_Rxs_sp_case}.
As with the proof of converse of Theorem~\ref{thm:3_Rxs_sp_case}, the converse relies on the first inequality in the statement of Theorem~\ref{thm:U_ESU}, the fact that any valid index code for $\BfullS$ has rate at least $\mais(\GfullS)$, the hypothesis $\mais(\GfullS)=\alpha(\GufullS)$, and the fact that $\GufullS$ is perfect, i.e., $\alpha(\GufullS) = \bar{\chi}(\GufullS)$.


\subsection*{Converse for the case $\mais(\GfullS) < \alpha(\GufullS)$}

The proof technique used in the converse of Theorem~\ref{thm:3_Rxs_sp_case} does not hold for the unicast problems where $\mais(\GfullS)<\alpha(\GufullS)$. 
Theorem~\ref{thm:mais_less_alpha} identifies these problems.  
The graph $\GfullS$ of one such problem is shown in Fig.~\ref{fig:special_side_graph}.
\begin{figure}[t]
    \centering
    \includegraphics[width=3in]{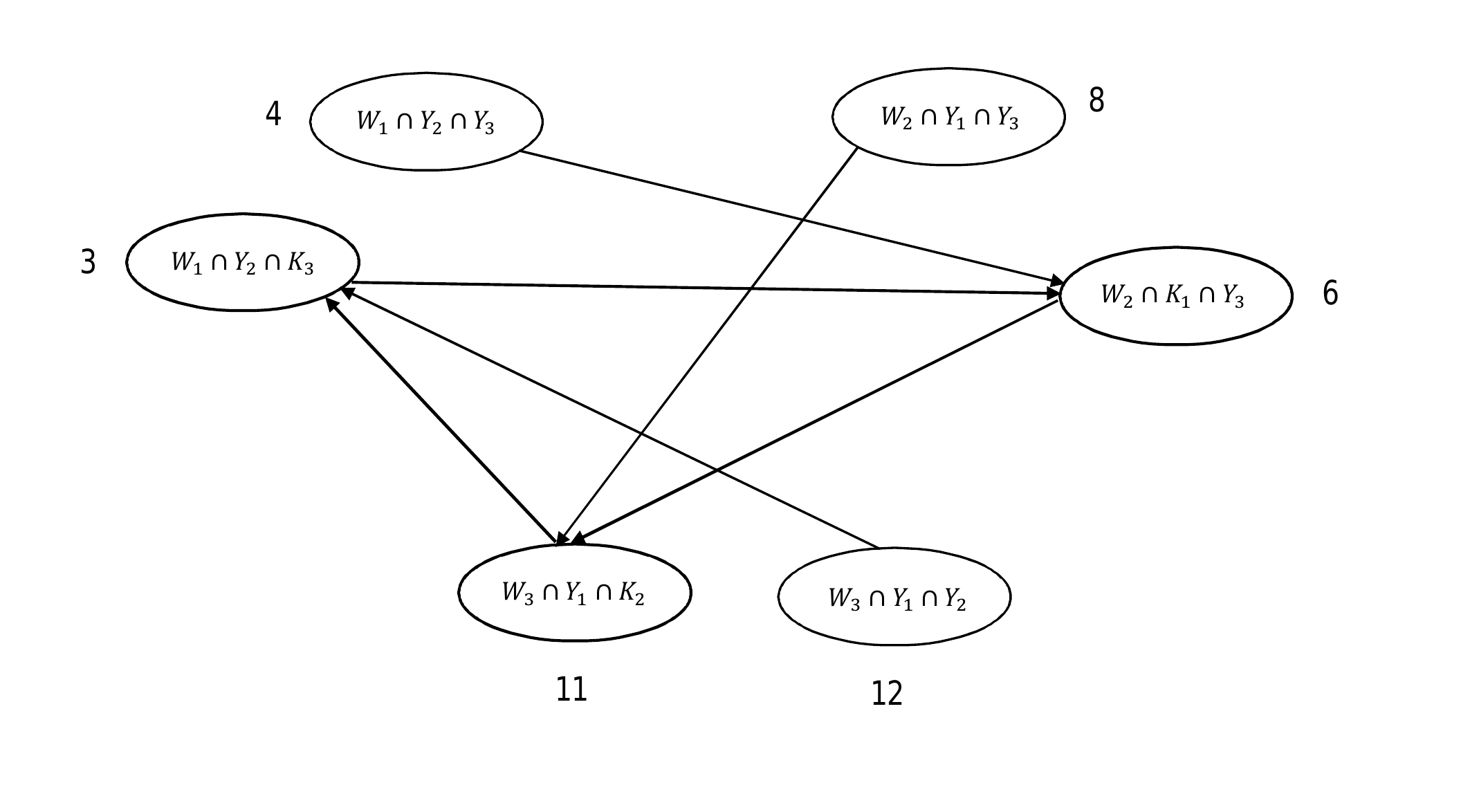}
    \vspace{-5mm}
    \caption{Side information graph $\GfullS$ for $S=\{3,6,11\} \cup \{4,8,12\}$.}
    \label{fig:special_side_graph}
\end{figure}
For these unicast problems, we use an information theoretic argument to find $\beta_{\BfullS}^*(1)$. 
The following lemma will be useful in proving the converse.

\begin{lemma} \label{lma:XYZU}
Let the random variables $X,Y,Z$ be independent, and let the random variable $U$ satisfy $H(U|X,Y) = H(U|X,Z) = 0$. Then $H(U|X) = 0$.
\end{lemma}
\begin{proof}
Since $H(U|X,Y) = 0$, we have $H(U|X,Y,Z)=0$. Hence, $I(U;Z|X,Y) = H(U|X,Y) - H(U|X,Y,Z) = 0$. Using the chain rule of mutual information
\begin{align*}
0 &= I(U;Z|X,Y) = H(Z|X,Y) - H(Z|X,Y,U) \\
  &= H(Z) - H(Z|X,Y,U)  ~~~~~(\text{since } X,Y,Z \text{ are ind.} )
\end{align*} 
We conclude from the above equation that $Z$ is independent of $(X,Y,U)$, and hence, $Z$ is independent of $(X,U)$.

Since $H(U|X,Z) = 0$, we have $H(U) = I(U;X,Z)$, i.e.,
\begin{align*}
0 &= I(U;X,Z) - H(U) = H(X,Z) - H(X,Z|U) - H(U) \\
  &= H(X) + H(Z) - H(X|U) - H(Z|X,U) - H(U) \\
  &= H(X) - H(X|U) - H(U)~~~(\text{since } Z, (X,U) \text{ are ind.} ) \\
  &= I(X;U) - H(U) = \,-H(U|X).
\end{align*} 
\end{proof}



Consider any $S$ such that \mbox{$\mais(\GfullS) < \alpha(\GufullS)$}. Let the index set of demands and side information of receiver $u_i$ in the problem $\BfullS$ be $W_i$ and $K_i$, respectively. 
Using the same indexing of messages as in Fig.~\ref{fig:side_graph}, we have $W_1=\{1,\dots,4\} \cap S$, $W_2 = \{5,\dots,8\} \cap S$, $W_3 = \{9,\dots,12\} \cap S$, $K_1 = \{5,6,9,10\} \cap S$, $K_2 = \{1,2,9,11\} \cap S$, $K_3 = \{1,3,5,7\} \cap S$.
Also, let $R_1,R_2,R_3$ be the index set of codeword symbols observed by the receivers $u_1,u_2,u_3$, respectively. 

\begin{lemma} \label{lma:intersection_R}
If $S$ is such that $\mais(\GfullS) < \alpha(\GufullS)$, then for any valid index code for $\BfullS$ with locality $1$, we have \mbox{$|R_i \cap R_j|=0$} for every choice of $1 \leq i < j \leq 3$.
\end{lemma}
\begin{proof}
Assume that the messages $\xb_j, j \in S$ are random, independent of each other and are uniformly distributed in $\Ac^m$. The logarithms used in measuring mutual information and entropy will be calculated to the base $|\Ac|$. Thus, $H(\xb_j)=m$, $H(\xb_{W_i}) = m|W_i|$ and $\xb_{W_i}$ and $\xb_{K_i}$ are independent of each other since $W_i \cap K_i = \phi$.

Let $\cb$ be the codeword generated by a valid index code for $\BfullS$ with $r=1$. For any $i \in [3]$, $u_i$ can decode $\xb_{W_i}$ from $\xb_{K_i}$ and $\cb_{R_i}$, and hence,
\begin{align}
I(\xb_{W_i};\cb_{R_i}|\p{x}_{K_i}) &= H(\xb_{W_i}) = m|W_i| \nonumber \\
&= H(\cb_{R_i}|\p{x}_{K_i}) - H(\cb_{R_i}|\p{x}_{K_i},\xb_{W_i}). \label{eq:thm:mais_less_alpha:2}
\end{align} 
Since $H(\cb_{R_i}|\p{x}_{K_i}) \leq |R_i| = m|W_i|$ and $H(\cb_{R_i}|\p{x}_{K_i},\xb_{W_i}) \geq 0$, we conclude from~\eqref{eq:thm:mais_less_alpha:2} that $H(\cb_{R_i}|\p{x}_{K_i},\xb_{W_i})=0$ and $H(\cb_{R_i}|\p{x}_{K_i}) = H(\cb_{R_i}) = |R_i| = m|W_i|$, i.e., $\cb_{R_i}$ is a function of $\xb_{K_i}$ and $\xb_{W_i}$, and $\cb_{R_i}$ is independent of $\xb_{K_i}$ and is uniformly distributed in $\Ac^{|R_i|}$.
For any $i \neq j$ and $i,j \in [3]$, $\cb_{R_i \cap R_j}$ is a sub-vector of $\cb_{R_i}$, and thus,
\begin{align}
H(\cb_{R_i\cap R_j}|\p{x}_{K_i}) &= H(\cb_{R_i \cap R_j}) = |R_i\cap R_j| \label{eq:entropy1} \\
H(\cb_{R_i\cap R_j}|\p{x}_{K_i},\xb_{W_i}) &= H(\cb_{R_i \cap R_j} | \xb_{K_i \cup W_i}) = 0 \label{eq:entropy10}
\end{align}
Similarly considering the decoding operation at $u_j$,
\begin{align}
H(\cb_{R_i\cap R_j}|\p{x}_{K_j}) &= |R_i\cap R_j| \label{eq:entropy2} \\
H(\cb_{R_i\cap R_j}|\p{x}_{K_j},\xb_{W_j}) &= H(\cb_{R_i \cap R_j}|\xb_{K_j \cup W_j}) = 0   \label{eq:entropy20}
\end{align}

We will now use Lemma~\ref{lma:XYZU} with~\eqref{eq:entropy10} and~\eqref{eq:entropy20} to proceed with the proof. For this, let $A = (K_i \cup W_i) \cap (K_j \cup W_j)$. Then,~\eqref{eq:entropy10} and~\eqref{eq:entropy20} can be rewritten as
\begin{align}
0 &= H(\cb_{R_i \cap R_j} | \xb_A,\xb_{(K_i \cup W_i) \setminus A}) \nonumber \\
  &= H(\cb_{R_i \cap R_j} | \xb_A,\xb_{(K_j \cup W_j) \setminus A})  \label{eq:entropy3}
\end{align} 
Since $A,(K_i \cup W_i) \setminus A,(K_j \cup W_j) \setminus A$ are non-intersecting, the random variables $\xb_A,\xb_{(K_i \cup W_i) \setminus A},\xb_{(K_j \cup W_j) \setminus A}$ are independent.
From Lemma~\ref{lma:XYZU} and~\eqref{eq:entropy3}, we deduce 
\begin{equation} \label{eq:entropy4}
H(\cb_{R_i \cap R_j} | \xb_A) = 0.
\end{equation} 
Note that 
\begin{equation} \label{eq:thm:mais_less_alpha:A}
A = (K_i \cap W_j) \cup (K_i \cap K_j) \cup (W_i \cap K_j) \cup (W_i \cap W_j).
\end{equation} 

If $W_i \cap K_j \neq \phi$ and $W_j \cap K_i \neq \phi$, then there exist two vertices in $\GufullS$ that are adjacent to each other, implying that $S$ is not an independent set in $\Gufull$.
However, Corollary~\ref{cor:mais_less_alpha} implies that $S$ is indeed independent in $\Gufull$. Thus, we conclude that at least one of $W_i \cap K_j$, $W_j \cap K_i$ is empty. Without loss of generality we will assume that $W_i \cap K_j = \phi$.
We also know $W_i\cap W_j =\phi$ since this is a unicast index coding problem. Using these facts in~\eqref{eq:thm:mais_less_alpha:A}, we conclude $A \subseteq K_i$. Combining this with~\eqref{eq:entropy4}, we have
\begin{equation*}
H( \cb_{R_i \cap R_j} | \xb_{K_i} ) = 0.
\end{equation*} 
This equality together with~\eqref{eq:entropy1} implies $|R_i \cap R_j| = 0$.
This completes the proof.
\end{proof}

We can now provide a lower bound on the rate of valid index codes with locality one, and complete the proof of the converse.

\begin{lemma}
If $S$ is such that \mbox{$\mais(\GfullS) < \alpha(\GufullS)$}, then we have $\beta_{\BfullS}^{*}(1) \geq \bar{\chi}(\GufullS)$.
\end{lemma}
\begin{proof}
Consider any valid index code with locality $1$ for $\BfullS$. Let the length of the code be $\ell$. 
Since $r=1$, we have $|R_i|=m|W_i|$ for $i \in [3]$. Using the fact $W_1,W_2,W_3$ are pairwise non-intersecting and the number of messages is $|S|$, we have $\sum_{i \in [3]}|W_i| = |S|$ and hence, $\sum_{i \in [3]}|R_i| = m|S|$.
Since $R_1,R_2,R_3 \subseteq [\ell]$, we have
\begin{align}
\ell &\geq \, |\! \cup_{i \in [3]} \! R_i| \nonumber \\
     &= |R_1| + |R_2| + |R_3| - \sum_{i \neq j} |R_i \cap R_j| + |R_1 \cap R_2 \cap R_3| \nonumber \\
     &= m|S| - \sum_{i \neq j} |R_i \cap R_j| + |R_1 \cap R_2 \cap R_3|. \label{eq:thm:mais_less_alpha}
\end{align} 
From Lemma~\ref{lma:intersection_R}, we know that $|R_i \cap R_j| = 0$ for any $i \neq j$. This implies $|R_1 \cap R_2 \cap R_3|=0$, and hence from~\eqref{eq:thm:mais_less_alpha}, we have $\ell \geq m|S|$, i.e., $\beta = \ell/m \geq |S|$.
From Corollary~\ref{cor:mais_less_alpha}, $|S|= \bar{\chi}(\GufullS)$. This completes the proof.
\end{proof}

\section{Conclusion}

We considered three receiver unicast index coding problems and obtained optimum broadcast rate for locality equal to one. Our future work will be on finding the optimum broadcast rates at larger values of localities.




\end{document}